\documentclass[journal]{IEEEtran}
\newcommand{\SortNoop}[1]{}
\usepackage{amssymb,stmaryrd,amsmath,amsfonts,rotating,mathrsfs}

\usepackage[noadjust]{cite}
\usepackage{color}
\usepackage{epic}
\RequirePackage{bbm}
\newtheorem{theorem}{Theorem}

\newcommand{\btheo}{\begin{theorem}}
\newcommand{\etheo}{\end{theorem}}
\newcommand{\bproof}{\begin{proof}}
\newcommand{\eproof}{\end{proof}}
\newtheorem{definition}[theorem]{Definition}
\newcommand{\bdefi}{\begin{definition}}
\newcommand{\edefi}{\end{definition}}
\newtheorem{fact}[theorem]{Fact}
\newcommand{\bprop}{\begin{fact}}
\newcommand{\eprop}{\end{fact}}
\newtheorem{corollary}[theorem]{Corollary}
\newcommand{\bcor}{\begin{corollary}}
\newcommand{\ecor}{\end{corollary}}
\newtheorem{example}[theorem]{Example}
\newcommand{\bex}{\begin{example}}
\newcommand{\eex}{\end{example}}
\newtheorem{lemma}[theorem]{Lemma}
\newcommand{\blemma}{\begin{lemma}}
\newcommand{\elemma}{\end{lemma}}
\newtheorem{remark}[theorem]{Remark}
\newcommand{\bremark}{\begin{remark}}
\newcommand{\eremark}{\end{remark}}
\newtheorem{conj}[theorem]{Conjecture}
\newcommand{\bconj}{\begin{conj}}
\newcommand{\econj}{\end{conj}}

\newcommand{\naturals}{\ensuremath{\mathbb{N}}}

\def\0{{\tt 0}} 
\def\1{{\tt 1}} 
\def\?{{\tt *}} 
\renewcommand{\mid}{\,|\,}

\newcommand{\EEx}{\hfill $\Diamond$}

\allowdisplaybreaks
\begin{document}
\title{Universal Bounds on the Scaling Behavior of Polar Codes}
\author{Ali Goli, S. Hamed Hassani and R\"udiger Urbanke \thanks{A. Goli is with department of Electrical Engineering, Sharif University of Technology, Iran. H. Hassani and R. Urbanke are with School of Computer
 \& Communication Sciences, EPFL, Switzerland.  The work of Hamed Hassani was supported by Swiss National 
Science Foundation Grant no 200021-121903.
}
}

\maketitle
\begin{abstract}
We consider the problem of determining the trade-off between the
rate and the block-length of polar codes for a given block error
probability when we use  the successive cancellation decoder. We
take the sum of the Bhattacharyya parameters as a proxy for the
block error probability, and show that there exists a universal
parameter $\mu$ such that for any binary memoryless symmetric channel
$W$ with capacity $I(W)$, reliable communication requires rates
that satisfy $R< I(W)-\alpha N^{-\frac{1}{\mu}}$, where $\alpha$
is a positive constant and $N$ is the block-length.  We provide
lower bounds on $\mu$, namely $\mu \geq 3.553$, and we conjecture
that indeed $\mu=3.627$, the parameter for the binary erasure
channel.  \end{abstract}

\section{Introduction}
Polar coding schemes provably achieve the capacity of a wide array of
 channels including binary memoryless symmetric (BMS) channels. Let $W$ be a BMS channel with 
capacity $I(W)$.
In \cite{Ari09}, Ar\i kan showed that  for any rate $R <I(W)$
 the block error probability of the successive cancellation (SC)
 decoder is upper bounded by $N^{-1/4}$ for block-length $N$ large enough.  In \cite{Art09}, Ar\i kan and Telatar
significantly tightened this result. They showed that for any rate $R <I(W)$
and any $\beta < \frac12$, the block error probability is upper bounded
by $2^{-N^{\beta}}$ for $N$ large enough. Later in \cite{HMTU}, 
these bounds were refined to be dependent on $R$ and it was shown that
similar  asymptotic lower bounds are valid when we perform MAP decoding. Hence, SC and MAP decoders share the same asymptotic performance in this sense.  
Such an exponential decay suggests that error floors should not
be a problem for polar codes even at moderate block lengths
(e.g. $N> 10^4$).

Another problem of interest in the area of polar codes 
is to determine the  trade-off between the rate and the block-length for a given error
probability when we use  the successive cancellation (SC) decoder.  
In other words, in order to have reliable transmission 
with block error probability at most $\epsilon$, how does 
the maximum possible rate $R$ scale in terms of the  block-length $N$? 
This problem has been previously considered
 in \cite{KMTU10} and \cite{HKU10} mainly for the family of Binary Erasure Channels (BEC). 
In both \cite{KMTU10} and \cite{HKU10}, the authors 
provide strong evidence (both numerically and analytically) 
that for polar codes with the SC decoder, reliable communication 
over the BEC requires rates $N^{-\frac{1}{\mu}}$ below capacity, where $\mu \approx 3.627$.

In this paper, we provide rigorous lower bounds on the value of
$\mu$, such that for any BMS channel $W$, reliable transmission (in
the sense that  the sum of the Bhattacharyya parameters is small)
requires rates at least $N^{-\frac{1}{\mu}}$ below capacity. We
begin by giving the notation and the general problem set-up.

\subsection{Periminilaries}
Let $W: \mathcal{X} \to \mathcal{Y}$ be a BMS channel, with input alphabet $\mathcal{X}
=\{0,1\}$,  output alphabet  $\mathcal{Y}$, and the transition probabilities $\{W(y \mid x): x \in \mathcal{X}, y \in \mathcal{Y}\}$. 
 We consider the following three parameters for the channel $W$, 
\begin{align}
& H(W)= \sum_{y \in \mathcal{Y}}W(y \mid 1) \log \frac{W(y \mid 1) + W(y\mid 0)}{W(y \mid 1)}, 
 \label{H(W)} \\
& Z(W)= \sum_{y \in \mathcal{Y}} \sqrt{W(y\mid 0)W(y \mid 1)},  \label{Z(W)}\\
 & E(W)=  \frac 12  \sum_{y \in \mathcal{Y}} W(y \mid 1) 
e^{-\frac 12( \ln \frac{W(y \mid1)}{W(y \mid 0)}+
 \mid \ln \frac { W(y \mid1)}{W(y \mid 0)} \mid)}. \label{E(W)}
\end{align}
The parameter $H(W)$ is equal to the entropy of the output 
of $W$ given its input when we assume  uniform distribution on 
the inputs, i.e.,  $H(W)=H(X\mid Y)$. Hence, we call the parameter $H(W)$ the
 entropy of the channel $W$.  Also note that the capacity of $W$, which we denote by $I(W)$, is given by $I(W)=1-H(W)$.  The parameter $Z(W)$ is called the Bhattacharyya parameter of $W$ and $E(W)$ is called the error probability of $W$. It can be shown that $E(W)$ is equal to the error probability 
in estimating the channel input $x$ on the basis of 
the channel output $y$ via the maximum-likelihood 
decoding of $W(y|x)$ (with the further assumption 
that the input has uniform distribution). It can be shown 
that the following relations hold between these parameters
 (see for e.g., \cite{Ari09} and \cite[Chapter 4]{RiU08}):
\begin{align} \label{bounds1}
& 0 \leq 2 E(W) \leq H(W) \leq Z(W) \leq 1,\\
\label{bounds2}
& H(W)\leq h_2(E(W)),\\
&  Z(W) \leq \sqrt{1-(1-H(W))^2},
\end{align}
where $h_2(\cdot)$ denotes the binary entropy function.

\subsection{Channel transform}
Let $\mathcal{W}$ denote the set of all the BMS channels and consider a transform $W \to (W^-, W^+)$ that maps $\mathcal{W}$ to $\mathcal{W}^2$ in the following manner.
Having the channel $W: \{0,1\} \to \cal Y$,  the channels $W^-: \{0,1\} \to {\cal Y}^2$ and $W^+: \{0,1\} \to \{0,1\} \times  {\cal Y}^2$  are
defined as
\begin{align} \label{1}
& W^-(y_1,y_2 | x_1)= \sum_{x_2 \in \{0,1\} } \frac 12 W(y_1| x_1 \oplus x_2) W(y_2|x_2)  \\ \label{2}
& W^+(y_1,y_2,x_1 | x_2)= \frac 12 W(y_1 | x_1 \oplus x_2) W(y_2 | x_2),
\end{align}
A direct consequence of the chain rule of entropy yields
\begin{equation} \label{I_preserve}
\frac{H(W^+)+ H(W^-)}{2}= H(W)
\end{equation}
One can also show that,
\begin{align} 
& \label{ex-}  H(W) \leq H(W^-) \leq 1-(1-H(W))^2, \\
& \label{ex+}  H(W)^2 \leq H(W^+) \leq H(W). 
\end{align}

\subsection{Polarization process}
Let $\{B_n, n\geq 1\}$ be a sequence of iid Bernoulli($\frac12$) 
random variables. Denote by $(\mathcal{F}, \Omega, \mathbb{P})$ the 
probability space generated by 
this sequence and let  $(\mathcal{F}_n, \Omega_n, \mathbb{P}_n)$ be the probability 
space generated by $(B_1, \cdots,B_n)$.  For a BMS channel $W$, define a 
random sequence of channels $W_n$, $n \in \naturals \triangleq \{0,1,2,\cdots\}$, 
as $W_0=W$ and
\begin{equation} \label{W_n}
W_{n}=  \left\{
\begin{array}{lr}
W_{n-1} ^{+} &  \text{If $B_n=1$},\\
W_{n-1} ^{-} &   \text{If $B_n=0$},
\end{array} \right.
\end{equation}
where the channels on the right side are given by the transform $W_{n-1}\to (W_{n-1}^-, W_{n-1}^+)$. 
Let us also define the random processes $\{H_n\}_{n \in \naturals}$, $\{I_n\}_{n \in \naturals}$ and $\{Z_n\}_{n \in \naturals}$ as
$H_n=H(W_n)$, $I_n=I(W_n)=1-H(W_n)$ and $Z_n=Z(W_n)$. From \eqref{I_preserve}
 one can easily observe that $H_n$ (and $I_n$) is a martingale with $\mathbb{E}[H_n]=H(W)$.
It is further known from \cite{Ari09} that the processes $H_n$ and $Z_n$ converge almost 
surely to limit random variables $H_\infty$ and $Z_\infty$ and furthermore, these limit 
random variables take their values in the set $\{0,1\}$ with $\text{Pr}(H_\infty=0)=
\text{Pr}(Z_\infty=0)=H(W)$. 

\subsection{Polar codes}
Given the rate $R<I(W)$, polar coding  is based on choosing a 
set of $2^nR$ rows of 
the matrix 
$G_n= \bigl [ \begin{smallmatrix} 1 & 0 \\ 1 &1  \end{smallmatrix} \bigr]^{\otimes n}$ 
to form a $2^nR \times 2^n$ matrix 
which is used as the generator matrix in the 
encoding 
procedure\footnote{There are extensions of 
polar codes given in \cite{KSU} which use different kinds of matrices.}.  
The way this set is chosen is dependent on the channel $W$ and is briefly explained as follows:  
At time $n \in \naturals$, consider a specific 
realization of the sequence $(B_1, \cdots, B_n)$, and denote it by $(b_1, \cdots, b_n)$. 
The random variable $W_n$ outputs a BMS channel, according to the procedure \eqref{W_n},
which we can naturally denote by 
$W^{(b_1,\cdots,b_n)}$. Let us now identify a sequence $(b_1, \cdots, b_n)$ by an integer 
$i$ in the set $\{1, \cdots,N\}$ such that the binary expansion of $i-1$ is equal to the 
sequence $(b_1, \cdots,b_n)$, with $b_1$ as the least significant bit. As an example
for $n=3$, we identify $(b_1,b_2,b_3)=(0,0,1)$ with $5$ and $(b_1,b_2,b_3)=(1,0,0)$ with $2$.
To simplify notation, we use $W_n^{(i)}$ to denote $W^{(b_1,\cdots,b_n)}$.  
Given the rate $R$, the indices of the matrix $G_n$ are chosen as follows: 
Choose a 
subset of size $NR$ from the set of channels $\{W_{N}^{(i)}\}_{1 \leq i \leq N}$  
that have the least possible error probability (given in \eqref{E(W)})
and choose the rows $G_n$ with the same indices as these channels. 
E.g., if the channel $W_{N}^{(j)}$ is chosen, then 
the $j$-th row of $G_n$ is selected. 
In the following, given $N$, we call the set of indices of $NR$ channels 
with the least error probability, the set of good indices and denote it by 
$\mathcal{I}_{N,R}$. 

It is proved in \cite{Ari09} that the block error probability of such polar coding scheme
 under SC decoding, denoted by $P_e(N,R)$, is bounded from both sides by\footnote{Note here that by \eqref{E(W)} the error probability of a BMS channel is less that its Bhattacharyya value. Hence, the right side of 
\eqref{P_e} is a better upper bound for the block error probability than the sum of the Bhattacharyya values.}
\begin{equation} \label{P_e}
\max_{i \in \mathcal{I}_{N,R} } E(W_N^{(i)}) \leq P_e(N,R) \leq \sum_{i \in \mathcal{I}_{N,R}} E(W_{N}^{(i)}).
\end{equation}   
\subsection{Main results}
Consider a BMS channel $W$ and let us assume that a  
polar code with block-error probability  at most a given
 value $\epsilon >0$, is required.  One way to accomplish this is to ensure 
that the right side of \eqref{P_e} is less than $\epsilon$. However, this is only a sufficient condition that 
might not be necessary. Hence, we call the right side of \eqref{P_e} 
\emph{the strong reliability condition}. Based on 
this measure of the block-error probability, we provide bounds on how the rate $R$ scales in terms of the block-length $N$.
\begin{theorem} \label{main}
For any BMS channel $W$ with capacity $I(W) \in (0,1)$, 
there exist constants $\epsilon,\alpha>0$, which depend only on $I(W)$, 
such that 
\begin{align} \label{sum-eps}
 \sum_{i \in \mathcal{I}_{N,R}} E(W_{N}^{(i)}) \leq \epsilon,
\end{align}
implies 
\begin{equation} \label{Rb}
R <I(W)- \frac{\alpha}{N^\frac{1}{\mu}},
\end{equation}
where $\mu$ is a universal parameter lower bounded by $3.553$.
\QED 
\end{theorem}
A few comments are in order: 

1) As we will see in the sequel, we can obtain an increasing sequence
of lower bounds, call this sequence $\{\mu_m\}_{m\in \naturals}$,
for the universal parameter $\mu$. For each $m$, in order to show
the validity of the lower bound we need to verify the concavity of
a certain polynomial (defined in \eqref{f_n}) in $[0, 1]$. For small
values of $m$ concavity can be proved directly using pen and paper.
For larger values of $m$ we can automate this process: each
polynomial has rational coefficients. Hence also its second derivative
has rational coefficients. To show concavity it suffices to show
that there are no roots of this second derivative in $[0, 1]$. This
task can be accomplished exactly by computing so-called Sturm chains
(see Sturm's Theorem \cite{sturm}). Computing Sturm chains is equivalent to
running Euclid's algorithm starting with the second and third
derivative of the original polynomial.  The lower bound for $\mu$
stated in Theorem~\ref{main} is the one corresponding to $m=8$, an
arbitrary choice.  If we increase $m$ we get e.g., $\mu_{16}=3.614$.
We conjecture that the sequence $\mu_m$ converges to $\mu=3.627$,
the parameter for the BEC.

2) Let $\epsilon,\alpha,\mu$ be as in Theorem~\ref{main}.  If we
require the block-error probability to be less than $\epsilon$ (in
the sense that the condition  \eqref{sum-eps} is fulfilled), then
the block-length $N$ should be at least \begin{equation} N >
(\frac{\alpha}{I(W)-R})^\mu.  \end{equation}

3) It is well known that the value of $\mu$ for the random linear ensemble is $\mu=2$, which is the optimal 
value since  the
variations of the channel itself require $\mu \geq 2$. Thus, given a block-length $N$, 
reliable transmission by polar codes requires a larger gap to the channel capacity than the optimal 
value.   

The rest of the paper is devoted to proving Theorem~\ref{main}. In Section~\ref{speed}, we provide universal 
lower bounds on how fast the process $H_n$ converges to its limit $H_\infty$. We then use these 
bounds to prove Theorem~\ref{main} in Section~\ref{proof}.
Finally, Section~\ref{open} concludes the paper with stating the related open questions. 

\section{Universal Lower bounds on the speed of polarization} \label{speed}
Consider a channel $W$ with its entropy process $H_n=H(W_n)$. Since the bounded process $H_n$ converges almost surely to a
$0-1$ valued random variable, we have $\lim_{n\to\infty}\mathbb{E}[H_n(1-H_n)]=0$. In this section, we  provide universal lower bounds on the speed 
with which the quantity $\mathbb{E}[H_n(1-H_n)]$ decays to $0$.  We first derive such lower bounds for the family of Binary Erasure Channels (BEC)
 and then extend them to other BMS channels. 
\subsection{Binary erasure channel}
Consider a binary erasure channel with erasure probability $h \in [0,1]$ which we denote by BEC($h$). One can show that (see \cite[Chapter 4]{RiU08} ) 
for such a channel we have
\begin{equation}
H(\text{BEC}(h))=Z(\text{BEC}(h))=2E(\text{BEC}(h))=h.
\end{equation}
 Furthermore, we have
 \begin{align*}
 & (\text{BEC}(h))^+=\text{BEC}(h^2),\\
 & (\text{BEC}(h))^-=\text{BEC}(1-(1-h)^2),
 \end{align*}
 both proved in \cite{Ari09}.
Hence, the processes
$H_n$ and $Z_n$ for BEC($h$) are equal and have a simple closed form expression as the following: Let $H_0= h$ and\footnote{Note that to simplify notation we have dropped the dependency of $H_n$ to its starting value $H_0=h$.}
\begin{equation} \label{Z_n}
H_{n} =\left\{ \begin{array}{cc} 
H_{n-1}^2,&\text{If } B_n=1,\\
1-(1-H_{n-1})^2,&\text{If } B_n=0.
\end{array}\right.
\end{equation}
Let us now define the sequence of functions $\{f_n(h) \}_{n \in \naturals }$ as $f_n:[0,1]\to[0,1]$ and for $h \in[0,1]$,
\begin{equation}
f_n(h)=\mathbb{E}[H_n(1-H_n)].
\end{equation}
Here, note that for $h \in [0,1]$ the value of $f_n(h)$ is a deterministic value 
that is dependent on the process $H_n$ with the starting value $H_0=h$. By 
using the recursive relation \eqref{Z_n}, one can easily deduce that
\begin{align} \label{f_n}
& f_0(h)=h(1-h), \\
& f_{n}(h)= \frac{f_{n-1}(h^2)+ f_{n-1}(1-(1-h)^2)}{2}. \nonumber
\end{align}
Let us also define a sequence of numbers $\{a_m\}_{m\in \naturals}$  as
\begin{align}
& a_m = \inf_{h\in[0,1]} \frac{f_{m+1}(h)}{f_m(h)}. \label{a_m}
\end{align}
\begin{remark} \label{rem_a}
One can compute the value of $a_m$ by finding the extreme points
of the function $\frac{f_{m+1}}{f_m}$ (i.e., finding the roots of
the polynomial $g_m={f'}_{m+1}f_m-f_{m+1}{f'}_m$) and checking which
one gives the global minimum.  Again, for small values e.g., $m=0,
1$, pen and paper suffice. For higher values of $m$ we can again
automatize the process: all these polynomials have rational
coefficients and therefore it is possible to determine the number
of real roots exactly and to determine their value to any desired
precision (by computing Sturm chains as mentioned earlier).  Hence,
we can find the value of $a_m$ to any desired precision.
Table~\ref{a_val} contains the numerical value of $a_m$ up to
precision $10^{-4}$ for $m\leq 8$. As the table shows, the values $a_m$ are 
increasing (see Lemma~\ref{a_bec}), and we conjecture that they converge to $2^{-\frac{1}{3.62713}}=0.8260$, the corresponding value for the 
channel BEC.
\EEx
\end{remark}
\begin{table}
\centering
\begin{tabular}{c c c c c c c c }
$m$ & $0$  & $2$  & $4$    & $6$  & $8$  \\
\hline
$a_m$ & $0.75$  & $0.7897$  & $0.8075$  & $0.8190$  & $0.8228$ \\
\hline
$\mu_m$ &$2.409$  & $2.935$  & $3.241$  & $3.471$ & $3.553$ \\
\end{tabular}
\caption{ }
\label{a_val}
\end{table}  
We now show that each of the values $a_m$ is a lower bound on the speed of decay of the sequence $f_n$.
\begin{lemma}\label{a_bec}
Fix $m\in \naturals$. For all $n \geq m$ and $h \in [0,1]$, we have
\begin{equation} \label{BEC_bound}
(a_m)^{n-m} f_{m} (h) \leq f_{n}(h) .
\end{equation}
Furthermore, the sequence $a_m$ is an increasing sequence.
\end{lemma}
\begin{proof}
The proof goes by induction on $n-m$. For $n-m=0$ the result is trivial. 
Now, assume that the relation \eqref{BEC_bound} holds for a  $n-m = k$, i.e., for $h\in[0,1]$ we have
\begin{equation} \label{BEC_m}
(a_m)^{k} f_{m} (h) \leq f_{m+k}(h) 
\end{equation}
 We show that \eqref{BEC_bound} is indeed true for  $k+1$ and $h\in[0,1]$.  We have
 \begin{align*}
 f_{m+k+1}(h) & \stackrel{(a)}{=} \frac{f_{m+k}(h^2) + f_{m+k}(1-(1-h)^2)}{2} \\
 & \stackrel{(b)}{\geq} \frac{(a_m)^k f_{m}(h^2) + (a_m)^k f_{m}(1-(1-h)^2)}{2} \\
 & = (a_m)^k f_{m+1}(h)  \\
 & = (a_m)^{k} \frac{f_{m+1}(h)}{f_m(h)} f_m(h) \\
& \geq  (a_m)^{k}  \bigl [\inf_{h \in [0,1]}  \frac{f_{m+1}(h)}{f_m(h)} \bigr] f_m(h) \\
&= (a_m)^{k+1} f_m(h). 
 \end{align*}
 Here, (a) follows from \eqref{f_n} and (b) follows from the left side inequality in \eqref{BEC_m},
 and hence the lemma is proved via induction.
\end{proof}
\subsection{BMS Channels}
For a BMS channel $W$, there is no simple $1$-dimensional recursion for  process $H_n$ as for BEC. 
However, by using \eqref{ex-} and \eqref{ex+},  one can give  bounds on how $H_n$ evolves.
In this section, we use the functions $\{f_n\}_{n\in \naturals}$ defined in \eqref{f_n} to provide universal lower
 bounds on the quantity $\mathbb{E}[H_n(1-H_n)]$. We start by introducing one further technical condition given as follows.
\begin{definition}
 We call an integer $m \in \naturals$ \emph{suitable} if the function $f_m(h)$, defined in \eqref{f_n}, is concave on $[0,1]$.
\end{definition}
\begin{remark} \label{rem_b}
For small values of $m$, i.e., $m\leq2$, it is easy to verify by hand that the function $f_m$ is concave.
As discussed previously, for larger values of $m$ we can use Sturm's theorem \cite{sturm}
and a computer algebra system to verify this claim. Note that
the polynomials $f_m$ have integer coefficients. Hence, all
the required computations can be done exactly. Unfortunately,
the degree of $f_m$ is $2^{m+1}$. We have checked up to $m=8$
that $f_m$ is concave and we conjecture
 that in fact this is true
for all $m\in \naturals$.
\EEx
\end{remark} 
In the rest of this section, we show that  for any BMS channel $W$, the value of $a_m$ is a 
lower bound on the speed of decay of $H_n$ provided that $m$ is a suitable integer. 
\begin{lemma} \label{BMS}
Let $m \in \naturals$ be a suitable integer and $W$ a BMS channel. We have for $n \geq m$ 
\begin{equation}
 \mathbb{E}[H_{n}(1-H_{n})] \geq (a_m)^{n-m} f_m(H(W)),
\end{equation}
where $a_m$ is given in \eqref{a_m}.
\end{lemma}
\begin{proof}
We use induction on $n-m$: For $n-m=0$ there is nothing to prove. Now, assume that the result of the lemma
 is correct for $n-m=k$. 
Hence, for any BMS channel $W$ with $H_n=H(W_n)$ we have
\begin{equation} \label{BMS_m}
 \mathbb{E}[H_{m+k}(1-H_{m+k})] \geq (a_m)^k f_m(H(W)).
\end{equation}
We now prove the lemma for $m-n=k+1$. For the BMS channel $W$, let us recall that the the transform
$(W \to (W^-,W^+))$ yields two channels $W^-$ and $W^+$ such that the 
relation \eqref{I_preserve} holds. Define the process $\{{(W^-)}_n, n \in \naturals\}$ as the channel process that 
starts with $(W^-)_0=W^-$ and evolves as in \eqref{W_n} similarly define $\{{(W^+)}_n, n \in \naturals\}$ similar with $(W^+)_0=W^+$
. Let us also define the two processes $H_n^-=H({(W^-)}_n)$ and $H_n^+=H({(W^+)}_n)$. We have,
\begin{align*}
& \mathbb{E}[H_{m+k+1}(1-H_{m+k+1})] \\
&\stackrel{(a)}{=} \frac{\mathbb{E}[H_{m+k}^-(1-H_{m+k}^-)] +\mathbb{E}[H_{m+k}^+(1-H_{m+k}^+)] }{2} \\
&\stackrel{(b)}{\geq} (a_m)^k \frac{ f_m(H(W^-)) + f_m(H(W^+))}{2} \\
& \stackrel{(c)}{\geq} (a_m)^k \frac{f_m( 1-(1-H(W))^2) +f_m( H(W)^2)}{2} \\
&\stackrel{(d)}{=}  (a_m)^k f_{m+1}( H(W)) \\
& =  (a_m)^k \frac{f_{m+1}( H(W))}{f_m(H(W))} f_m(H(W)) \\
& \geq   (a_m)^k \bigl [ \inf_{h\in[0,1]}  \frac{f_{m+1}(h)}{f_m(h)} \bigr] f_m(H(W)) \\
& \stackrel{(e)}{=} (a_m)^{m+1}    f_m(H(W)).
\end{align*}  
In the above chain of inequalities, relation (a) follows from the fact that $W_m$ has $2^m$ possible outputs among which half of 
them are branched out from $W^+$ and the other half are
 branched out from $W^-$ . 
 Relation (b) follows from the induction hypothesis 
given in \eqref{BMS_m}. Relation (c) follows from \eqref{ex-}, \eqref{ex+} and 
the fact that the function $f_m$ is concave. More precisely, since $f_m$ is
 concave on $[0,1]$, we have the following inequality for any sequence of 
numbers $0 \leq x' \leq x \leq y \leq y' \leq 1$ that satisfy 
 $\frac{x+y}{2} = \frac{x'+y'}{2}$:  
\begin{equation} \label{conc}
 \frac{f_m(x') + f_m(y')}{2} \leq \frac{f_m(x) + f_m(y)}{2}.
\end{equation}  
In particular, we set $x'=H(W)^2$, $x=H(W^+)$, $y=H(W^-)$, $y'=1-(1-H(W))^2$ and we know from \eqref{ex-} and \eqref{ex+} that
 $0 \leq x' \leq x \leq y \leq y' \leq 1$. Hence, by \eqref{conc} we obtain (c). Relation (d) follows 
from the recursive definition of $f_m$ given in \eqref{f_n}. Finally, 
 relation (e) follows from the definition of $a_m$ given in \eqref{a_m}.
\end{proof}

\section{Proof of Theorem~\ref{main}} \label{proof}
To fit the bounds of Section~\ref{speed} into the framework of Theorem~\ref{main}, let us first introduce the sequence
 $\{\mu_m\}_{m \in \naturals}$ as
\begin{equation}
 \mu_m= -\frac{1}{\log a_m},
\end{equation}
where $a_m$ is defined in \eqref{a_m}.
In the last section, we proved that for a suitable $m$, the speed with which the quantity 
$\mathbb{E}[H_n(1-H_n)]$ decays is lower bounded by $a_m=2^{-\frac{1}{\mu_m}}$, i.e. for 
$n\geq m$ we have $\mathbb{E}[H_n(1-H_n)] \geq 2^{-\frac{(n-m)}{\mu_m}} f_m(H(W))$.
To relate the strong reliability condition in \eqref{sum-eps} to the rate bound in \eqref{Rb}, we need the following lemma. 
\begin{lemma} \label{gamma}
Consider a BMS channel $W$ and assume that there exist positive
 real numbers $\gamma, \theta$ and $m\in \naturals$ such that $\mathbb{E}[H_n(1-H_n)] \geq \gamma 2^{-n \theta}$ for $n \geq m$. 
Let $\alpha, \beta \geq 0$ be such that $2\alpha+ \beta=\gamma$, we have for $n \geq m$
\begin{equation}
\text{Pr}(H_n \leq \alpha 2^{-n \theta}) \leq I(W)-\beta2^{-n\theta}.
\end{equation}
\end{lemma}
\begin{proof}
The proof is by contradiction. Let us assume the contrary, i.e., we assume there exists $n\geq m$ s.t.,
\begin{equation} \label{contr}
\text{Pr}(H_n \leq \alpha 2^{-n \theta}) > I(W)-\beta2^{-n\theta}.
\end{equation}
 In the following, we show that with such an assumption we reach to a contradiction.  We have
 \begin{align}\nonumber
 &\mathbb{E}[H_n(1-H_n)] \\ \nonumber
 &= \mathbb{E}[H_n(1-H_n) \mid H_n \leq \alpha 2^{-n\theta}] \text{Pr}(H_n \leq \alpha 2^{-n\theta}) \\
 & \; \; \: \; +  \mathbb{E}[H_n(1-H_n) \mid H_n > \alpha 2^{-n\theta}] \text{Pr}(H_n > \alpha 2^{-n\theta}). \label{cont1}
 \end{align}
 It is now easy to see that 
 \begin{align*}
 \mathbb{E}[H_n(1-H_n) \mid H_n \leq \alpha 2^{-n\theta}]  \leq \alpha 2^{-n\theta},
\end{align*}
and since $\mathbb{E}[H_n(1-H_n)] \geq \gamma 2^{-n\theta}$, by using \eqref{cont1} we get
 \begin{equation} \label{cont2}
 \mathbb{E}[H_n(1-H_n) \mid H_n > \alpha 2^{-n\theta}] \text{Pr}(H_n > \alpha 2^{-n\theta}) \geq 2^{-n\theta}(\gamma-\alpha).
 \end{equation}
  We can further write
  \begin{align}\nonumber
 \mathbb{E}[(1-H_n)] &= \mathbb{E}[1-H_n \mid H_n \leq \alpha 2^{-n\theta}] \text{Pr}(H_n \leq \alpha 2^{-n\theta}) \\
 & \; \; \: \; +  \mathbb{E}[1-H_n \mid H_n > \alpha 2^{-n\theta}] \text{Pr}(H_n > \alpha 2^{-n\theta}), \label{cont3}
 \end{align}
 and by noting the fact that $H_n \geq H_n (1-H_n)$ we can plug in \eqref{cont2} in \eqref{cont3} to obtain
 \begin{align} 
 \mathbb{E}[(1-H_n)]  & \geq  \mathbb{E}[1-H_n \mid H_n \leq \alpha 2^{-n\theta}] \text{Pr}(H_n \leq \alpha 2^{-n\theta}) \nonumber \\ 
  & \;\;\; +2^{-n\theta}(\gamma-\alpha). \label{cont4}
 \end{align}
 We now continue by using \eqref{contr} in \eqref{cont4} to obtain
 \begin{align*}
 \mathbb{E}[(1-H_n)]  &> (I(W)-\beta2^{-n\theta})(1-\alpha2^{-n\theta})+ 2^{-n\theta}(\gamma-\alpha) \\
 &\geq I(W)+2^{-n \theta} (\gamma- \alpha (1+ I(W))-\beta),
 \end{align*} 
 and since $2\alpha+ \beta=\gamma$, we get $\mathbb{E}[1-H_n]>I(W)$. However, this is a contradiction
 since $H_n$ is a martingale and $\mathbb{E}[1-H_n]=I(W)$.
\end{proof}
Let us now use the result of Lemma~\ref{gamma} to conclude the proof of Theorem~\ref{main}.
By Lemma~\ref{BMS}, we have for  $n \geq m$
\begin{align*}
\mathbb{E}[H_n(1-H_n)] & \geq 2^{-\frac{(n-m)}{\mu_m}} f_m(H(W)) \\
&= 2^{-\frac{n }{\mu_m}} ( 2^{\frac{m}{\mu_m}} f_m (H(W))). 
\end{align*}   
Thus, if we now let 
\begin{align*}
&\gamma=  2^{\frac{m}{\mu_m}} f_m (H(W)), \\
&2\alpha=\beta=\frac{\gamma}{2},
\end{align*}
then by using Lemma~\ref{gamma} we obtain
\begin{equation} \label{ineqtm}
\text{Pr}(H_n \leq \frac{\gamma}{4} 2^{-\frac{n }{\mu_m}}) \leq I(W)-\frac{\gamma}{2}2^{-\frac{n }{\mu_m}}.
\end{equation}
Now, assume we desire to achieve a rate $R$ equal to 
\begin{equation} \label{Rm}
R= I(W)- \frac{\gamma}{4}2^{-\frac{n }{\mu_m}}. 
\end{equation}
Let $\mathcal{I}_{N,R}$ be the set of indices chosen for such a rate $R$, i.e., $\mathcal{I}_{N,R}$ includes the $2^nR$ 
indices of the sub-channels with the least value of error probability.  Define the set $A$ as
\begin{align}
A= \{ i \in \mathcal{I}_{n,R}: H(W_N^{(i)}) \geq \frac{\gamma}{4} 2^{-\frac{n }{\mu_m}}\}.
\end{align}
 In this regard, note that \eqref{ineqtm} and \eqref{Rm} imply that 
\begin{equation}
\mid A \mid \geq \frac{\gamma}{4} 2^{n(1-\frac{1}{\mu_m})},
\end{equation}
and as a result, by using \eqref{bounds1} and \eqref{bounds2}  we obtain
\begin{align*} \label{sum1}
\sum_{i \in \mathcal{I}_{N,R}} E(W_N^{(i)}) \geq \sum_{i \in A} E(W_N^{(i)}) &\geq \frac{\gamma^2}{16} 2^{n(1-\frac{1}{\mu_m})} h_2^{-1}(2^{-\frac{n }{\mu_m}})\\
& \geq  \frac{\gamma^2}{16} \frac{2^{n(1-2\ \frac{1}{\mu_m} )}}{8 n\ \frac{1}{\mu_m} },
\end{align*}
 where the last step follows from the fact that for $x \in [0,\frac{1}{\sqrt{2}}]$, we have $h_2^{-1}(x) \geq \frac{x}{8 \log(\frac{1}{x})}$. Thus, having 
a block-length $N=2^n$, in order to get to block-error probability (measured by \eqref{P_e})
  less than   $ \frac{\gamma^2}{16} \frac{2^{n(1-2\ \frac{1}{\mu_m} )}}{8 n\ \frac{1}{\mu_m} }$, the rate can be at
 most $R=I(W)-\frac{\gamma}{4}2^{-\frac{n }{\mu_m}}$. 

Finally, if we let $m=8$ (by the discussion in Remark~\ref{rem_b}, we know that $m=8$ is suitable), 
then $\mu_8=\frac{1}{-\log(a_8)}=3.553$ and choosing
\begin{equation}
\epsilon= \inf_{n \in \naturals} \bigl [ \sum_{i \in \mathcal{I}_{N,R}} E(W_N^{(i)}) \bigr ], 
\end{equation}
where $R$ is given in \eqref{Rm}, then we know for sure that $\epsilon>0$ (since $\frac{1}{\mu_8}>2$) and furthermore, 
to have block-error probability less that $\epsilon$ the rate should be less than $R$ given in \eqref{Rm}.
\section{Open problems} \label{open}
 The results of this paper can be
extended in the following ways.

1) In this paper, we take the right side of \eqref{P_e} as a proxy for the block error probability and hence 
our results are with respect to the strong reliability condition \eqref{sum-eps}. A significant step 
in this regard would be to prove equivalent bounds 
for the block error probability. 

2) Another way to improve the results of this paper is to provide better values of the universal parameter $\mu$. Based on numerical 
experiments, we conjecture that the value of $\mu$ can be increased up to the scaling parameter of the channel BEC. That is, 
the right value of 
$\mu$ to plug in \eqref{Rb} is equal to $\mu=3.62713$. Thus, the ultimate goal would 
be to show that for the channel BEC, the polarization phenomenon takes place
faster than all the other BMS channels. One way to do this, is to prove that the
 functions $f_n$ defined in \eqref{f_n} are concave on 
the interval $[0,1]$. 

3) The result of Theorem~\ref{main} suggests that in terms of finite-length performance, polar codes are far from optimal. 
However, we might get different results if we consider extended polar codes with $\ell \times \ell$ kernels (\cite{KSU}). It is not very hard 
to prove that at least for the BEC,  as $\ell$ grows large, for almost all the $\ell \times \ell$ kernels 
the finite-length performance of polar codes improves towards the optimal one 
(i.e., $\mu \to 2$). However, this is  at the cost of an increase in complexity proportional to $2^\ell$. 
This suggests that 
there might still exist kernels with reasonable size with superior finite-length properties than the original $2 \times 2$ kernel.  
Hence, an interesting open problem is the finite-length  analysis of polar codes that 
are constructed from $\ell \times \ell$ kernels and relate such analysis to finding kernels with better 
finite-length properties.

\section*{Acknowledgment}
The authors wish to thank anonymous reviewers for their
valuable comments on an earlier version of this
manuscript. 


\end{document}